\documentclass[letterpaper,11pt,reqno]{amsart}

\makeatletter
\usepackage{amssymb}
\usepackage{latexsym}
\usepackage{amsbsy}
\usepackage{amsfonts}
\usepackage{hyperref}
\usepackage{caption}
\usepackage{subcaption}
\usepackage{graphicx}
\usepackage{enumerate}
\usepackage{enumitem}
\usepackage{mathtools}
\usepackage{color}

\usepackage{tikz}

\def\marginpar#1{\ignorespaces}

\newtheorem{theorem}{Theorem}[section]
\newtheorem{lemma}[theorem]{Lemma}
\newtheorem{proposition}[theorem]{Proposition}

\newtheorem{definition}[theorem]{Definition}
\newtheorem{assump}[theorem]{Assumption}

\numberwithin{equation}{section}
\makeatother
\begin{document}
\title[Transaction fee mechanism for PoS]{Transaction Fee Mechanism for \\ Proof-of-Stake Protocol}

\author[Wenpin Tang]{{Wenpin} Tang}
\address{Department of Industrial Engineer and Operations Research, Columbia University. 
} \email{wt2319@columbia.edu}

\author[David Yao]{David D.\ Yao}
\address{Department of Industrial Engineer and Operations Research, Columbia University. 
} \email{yao@columbia.edu}

\date{\today} 
\begin{abstract}
We study a mechanism design problem in the blockchain proof-of-stake (PoS) protocol.
Our main objective is to extend the transaction fee mechanism (TFM) recently
proposed in \cite{CS23}, so as to incorporate
a {\it long-run} utility model for the miner into  
the burning second-price auction mechanism $\texttt{BSP}(\gamma)$ proposed in \cite{CS23}
(where $\gamma$ is a key parameter in the strict $\gamma$-utility model that is applied to both miners and users).
First, we derive an explicit functional form for the long-run utility of the miner using a martingale 
approach, and reveal a critical discontinuity of the utility function, namely a small deviation from being truthful
will yield a discrete jump (up or down) in the miner's utility.
We show that because of this discontinuity the $\texttt{BSP}(\gamma)$ mechanism will
fail a key desired property in TFM,  $c$-side contract proofness ($c$-SCP). 
As a remedy, we introduce another parameter $\theta$, and propose
a new $\texttt{BSP}(\theta)$ mechanism, and prove that it
satisfies all three desired properties  of TFM: user- and miner-incentive compatibility (UIC and MIC) as well as $c$-SCP,
 provided the parameter $\theta$ falls into
a specific range, along with a proper ``tick'' size imposed on user bids.
\end{abstract}

\maketitle

\textit{Key words}: Blockchain, proof of stake, transaction fee mechanism, cryptocurrency, incentive compatibility, utility, martingale.


\section{Introduction}

\quad 
A blockchain is a digital ledger that facilitates the secure exchange and execution of transactions 
in a distributed network without an intermediary, 
hence achieving {\em decentralization}.
The past decade has witnessed impressive advances of the 
blockchain technology in a wide range of applications including
cryptocurrency \cite{Naka08, Wood14},
healthcare \cite{MC19, TPE20},
supply chain \cite{CTT20, SK19},
non-fungible tokens \cite{Dow22, WL21}, 
and (more recently) crypto exchanges for the stock market \cite{MN23}, among many others.

\quad 
There are two primary parties in a blockchain, the users and the miners. 
Below is a brief highlight of the two parties' activities and interactions.
\begin{itemize}[itemsep = 3 pt]
\item
The users, who submit transactions to the blockchain for processing, seek to have their transactions 
 settled and published on the blockchain network in a timely fashion by the miners. 
Since each block has a limited size or capacity, 
most blockchains adopt an auction mechanism that requires the users to submit bids 
to have their transactions processed by the miners. 
In this regard, the blockchain {transaction} is similar to an auction system, with the miners acting like the auctioneer, and  
the users as bidders.
\item
The miners select a subset of transactions from the mempool (according to the bids),
include them in a block, and then position the block into an ever-growing public ledger.
At the core of this ``mining'' process is a {\em consensus protocol}, 
a set of rules governing the whole process, including both the selection of the miner and 
the detailed mechanics {(e.g., the ``longest chain")} to form the public ledger. 
The most popular protocols are {\em Proof of Work} (PoW) \cite{Naka08} and {\em Proof of Stake} (PoS) \cite{KN12, Wood14}.
In both protocols, the miners compete with each other (to be selected to do the work) by
either solving a hashing puzzle (PoW) or bidding with their stakes/coins (PoS).
The winner will be selected to attach a new block (i.e. mine the block) to the blockchain.
The selected miner will then receive transaction fees from the users,
along with a separate reward from the blockchain.
\end{itemize}
See Figure \ref{fig:0} for an illustration of the miner-user activities under the PoS protocol, which is the focus of this study.
\begin{figure}[h]
\centering
\includegraphics[width=0.74\columnwidth]{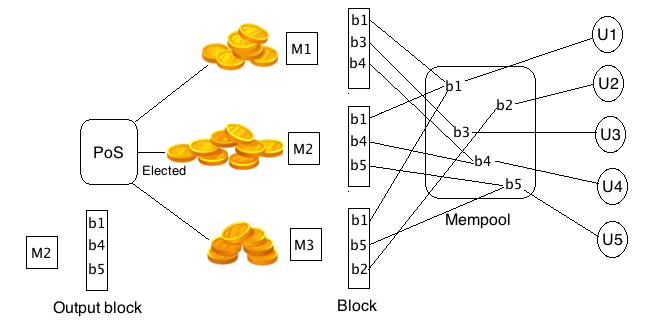}
\caption{Miner-user activities under the PoS protocol.}
\label{fig:0}
\end{figure}

\quad A central design issue for the PoS protocol is to come up with a reliable and efficient {\em transaction fee mechanism} (TFM)
so as to incentivize all participants, users and miners, {to} act honestly or truthfully. 
Most work in the general area of mechanism design (see \cite{Kris02, Myer81}) 
assumes that the auctioneer will honestly follow the prescribed mechanism,
so the only concern is the strategic behavior of bidders. 
In blockchains, however,
the miners (similar to the auctioneer) may also deviate from being truthful
based on ex-post information.
For instance, a miner may inject fake transactions so as to collect a larger payment, or collude with some users
(possibly after examining all the bids) so as 
to improve their joint utility.
What's more, \cite{BEOS19} showed that 
the Vickrey auction \cite{Vic61}, 
known to be user incentive compatible,
fails to be incentive compatible for the miners.

\quad To address the above problems, 
\cite{CS23, LSZ22, Rou20} proposed three desired properties 
for a transaction fee mechanism:
\begin{itemize}[itemsep = 3 pt]
\item
{\em User incentive compatibility} (UIC):
the users bid truthfully given that the miners implement the mechanism honestly.
\item
{\em Miner incentive compatibility} (MIC):
the miners follow the prescribed mechanism honestly. 
\item
{\em Miner-user side contract proofness} ($c$-SCP): 
any collusion between the miner and up to $c$ users cannot improve their joint utility 
by deviating from being truthful. 
\end{itemize}
Refer to Section \ref{sc33} for more details. 
See also \cite{AL20, CSZZ22, FW20, WSC23} for discussions
on (how to circumvent) the miner or auctioneer's strategic behaviors.
Recently, \cite{DG20} proposed the notion of {\em maximal extractable value} (MEV) 
to formally quantify and analyze the strategic behaviors in the blockchain generation, 
along with {the} so-called {\em proposer-builder separation} (PBS) \cite{BV21}
as an early attempt to reduce the miner's MEV.

\quad In 2019, Ethereum proposed a {TFM} EIP-1559 \cite{BCD19}, which
achieves all three properties above, as shown in  \cite{Rou20}, under the condition that 
there is no congestion, i.e. the block size is unlimited, while leaving open (unanswered) the congestion case.
A recent breakthrough \cite{CS23} provides a negative answer to this: 
when the block size is limited, as it is in practice, 
no {TFM} can achieve UIC, MIC and $c$-SCP simultaneously
under the so-called {\em current utility}, i.e., the utility applies only to the immediate gain/loss. 
Moreover, UIC and $c$-SCP implies that the miner's revenue must be zero.
Note, however, that in a less restrictive setting (“discrete” bids and unlimited block size), 
\cite{GY22} provides a TFM that can achieve all the three properties, with a positive miner's revenue under the current utility.
Specifically, a miner or a user may submit untruthful or fake bids, which may not be included in
the current block (and hence cost nothing) but, once submitted and registered in the blockchain
may still be included in a (possibly near) future block. This strategic behavior is not accounted for 
in the current utility model.


\quad As a remedy,  \cite{CS23} proposes a {\em strict $\gamma$-utility},
where the parameter $\gamma>0$ (strictly positive) serves as a discount factor that 
will be applied to any unconfirmed bids. 
Based on this, \cite{CS23} also proposes a randomized {TFM},
called the {\em burning second-price auction}, and shows it 
satisfies all the three properties under the strict $\gamma$-utility.
The essence of the strict $\gamma$-utility is to
account for the future cost of any strategic gain, and thus
mitigating the myopia  in the current utility (via neglecting the future cost). 

\quad The objective of our work here is to further generalize this idea,
by taking into account not only the future cost but also the difference between a miner's
perspective of the ``future" and that of a user's, particularly in terms of {\it time scale}. 
Any user, by definition a minor player in the blockchain system,
may only use the blockchain occasionally, to request the execution of a transaction (e.g., a buy/sell order). 
So it is reasonable to assume that the users primarily focus on the current or near-term interests in terms of 
 their utilities. 
In contrast, the miners are the major players performing the dual role of verifying the transactions
and updating the blockchain 
(refer to Figure \ref{fig:0}). 
Thus, a miner's utility is long-term based; and this is even more so in the PoS protocol, where
the probability for a miner to be selected to mine a new block 
(and earn the associated transaction payment and blockchain reward)
is proportional to the number of stakes that the miner can afford to put forth to bid
{(see \cite{RS21, Tang22, TY23}, and \cite{Tang23} for a review on the PoS wealth evolution.)}
Thus, the more profit (or stakes) a miner can accumulate over time, 
the more likely is the miner's chance to win the bid to mine the next block 
and thereby make yet more profit.  
In this regard, the miner's time scale, in terms of formulating a proper miner's utility model,
 should be very different from that of a typical user's. 

\quad In this paper, our main objective is to extend the TFM of \cite{CS23}, so as to incorporate
a long-run utility model for the miner into  
the burning second-price auction mechanism of \cite{CS23}, denoted $\texttt{BSP}(\gamma)$, where
$\gamma$ is the parameter (discount factor) in the strict $\gamma$-utility model.
In particular, we bring out a critical feature of such a long-run utility model (for the miner), a discontinuity  
in the sense that a small deviation from being truthful
will yield a discrete jump, either up or down, in the miner's utility (Theorem \ref{thm:Ump}).
This discontinuity calls for a substantial modification of the $\texttt{BSP}(\gamma)$ mechanism.
Specifically, we show the $\texttt{BSP}(\gamma)$ mechanism
cannot achieve $c$-SCP due to the discontinuity of the miner's utility (Proposition \ref{prop:cscp}).

\quad To overcome this handicap, we need to introduce another parameter $\theta$, which characterizes the randomized
confirmation rule (which selects a subset of top bids to be confirmed for inclusion into a new block). In addition, 
a minimal ``tick'' size $\Delta$ needs to be imposed on every bid submitted by the users, i.e., a bid can only be multiples of $\Delta$,
which bears similarity with the ``discrete" bids in \cite{GY22}.
Based on these two devices, we propose a new $\texttt{BSP}(\theta)$ mechanism, and prove that it
satisfies all three desired properties, UIC, MIC and $c$-SCP, under the PoS protocol, provided the parameter $\theta$ falls into
a specific range, along with a proper tick size $\Delta$.
Refer to details in Theorem \ref{thm:main}.

\quad The rest of the paper is organized as follows. 
In Section \ref{sc3}, we provide background and preliminary materials on the TFM 
in the context of the PoS protocol, define the utility functions and discuss the relevant strategic behaviors. 
In particular, in Section \ref{sc2}, we define the long-run utility for the miner mentioned above, explicitly derive its functional form, 
and reveal a crucial discontinuity in the miner's revenue (i.e., payment collected from the users whose bids are confirmed).
In Section \ref{sc33}, 
we propose a new burning second-price auction mechanism, denoted $\texttt{BSP}(\theta)$ as highlighted above.
We then show 
in Section \ref{sc4}, {that} the $\texttt{BSP}(\theta)$, along
with a minimum tick size and a proper parametric range for $\theta$, achieves
all three properties, UIC, MIC and $c$-SCP, for the PoS protocol.
We conclude with Section \ref{sc5}.

\section{Transaction fee mechanism, strategy and utility}
\label{sc3}


\quad In a transaction fee mechanism, 
a miner acts like an auctioneer, while users will bid 
to have their transactions included and confirmed in a block
and published in the blockchain. 

\quad Let $B$ be the number of slots in a block (i.e., the block size),
and assume without loss of generality that there are more bids than slots. 
The mechanism operates under the following rules:
\begin{itemize}[itemsep = 3 pt]
\item
An {\em inclusion rule} (executed by the miner) that decides 
which of the bids to include in the block. Only included bids can be accessed by the blockchain.
\item
A {\em confirmation rule} (executed by the blockchain) that 
 selects a subset of the included bids to confirm. 
Only confirmed bids are considered final, i.e., settled transactions.
\item
A {\em payment rule} (executed by the blockchain) that 
 specifies how much each confirmed bid should pay. 
\item
A {\em (miner) revenue rule} (executed by the blockchain) that 
specifies how much the miner should be paid. 
\end{itemize}

\quad There are several important facts to keep in mind.
To start with, only the first rule above, the inclusion rule, involves human decisions -- 
from the users (how much to bid) and the miners (which bids to include);
the other three rules concerning confirmation, user payment and miner revenue are all 
hard-coded into and executed by the blockchain protocol. 
The specification of these rules (all four) constitutes the main task of the mechanism design.

\quad Notably, however, since the inclusion rule interfaces with human decisions, the design of all
four rules need to anticipate and account for possible strategic behaviors of both the users and the miners.
Thus, in addition to the inclusion rule (executed by the miners), there's the confirmation rule (executed by the system),
i.e., not all bids included by a miner will be confirmed.
The other two rules specify how much a confirmed bid should pay, which need not be the same as what the user originally bids   
(e.g., similar to the second-price auction mechanism); and how much the miner should receive, which need not be the sum total of the payments 
from all confirmed bids.  
More detailed discussions are provided below. 


\subsection{Strategic behaviors and utilities}
\label{sc31}

As mentioned above, a good mechanism design is supposed to steer all participants  
 away from strategic (i.e., dishonest) behavior, so that they stay with their true (i.e., honest) values.
There are five sources of strategic behavior in the PoS protocol: those originate from
(i) a user,
(ii) a miner,
(iii) miner-user collusion,
(iv) user-user collusion;
and (v) miner-miner collusion.

\quad The impact of (iv), user-user collusion, is negligible, since the users are minor players in a blockchain network, and
it is difficult for them to collude.
So, below we shall ignore (iv).
As to (v), observe that under the PoS protocol, 
a miner's strategic decision will depend a priori on the number of stakes the miner owns, which determines the miner's chance to 
be selected to process the transaction.  
Thus, a miner-miner collusion can be reduced to the strategic behavior of a single representative miner,
who owns all the stakes (i.e., of the entire collusion group). 
So, below we shall implicitly assume there's a single (``super'') miner in the system.

\quad 
Specifically, we will consider the following possible strategic behaviors of the miner and the users:
\begin{itemize}[itemsep = 3 pt]
\item
{\em Strategically forming the inclusion list}.
A strategic miner or a miner-user collusion may not follow the inclusion rule,
as long as the included bids satisfy the block validity rules enforced by the blockchain.
\item
{\em Injecting fake transactions}.
The miner, users or a miner-user collusion may inject fake transactions, 
possibly after examining other users' bids in the mempool.
\item
{\em Bidding untruthfully}. 
The users or a miner-user collusion can bid untruthfully, possibly after examining other users' bids in the mempool.
\end{itemize}
Specifically, we shall adopt an {\em ex-post} auction (see \cite{Riley88})
in the transaction fee mechanism detailed below.
To do so, we need to first specify the utilities of the miner and the users. 

\begin{definition}[Strict $\gamma$-utility for the PoS]
\label{def:gamut}
For a bid $b$ (real or fake), denote by $v$ and $p$ its true value and required payment.
Let $\gamma\in (0,1]$ be a (given) parameter.
\begin{itemize}[itemsep = 3 pt]
\item
The strict $\gamma$-utility of the user is defined as the expected value of 
\begin{equation}
\label{eq:gammauser}
\sum_{b {\tiny \mbox{ confirmed}}} (v- p) - \sum_{b {\tiny \mbox{ unconfirmed}}, \, b > v} \gamma (b-v).
\end{equation}
\item
The strict $\gamma$-{\bf return} of the miner, denoted by $\mathcal{R}_\gamma$, 
is defined as the expected value of
\begin{equation}
\label{eq:gammaminer1}
\mbox{miner's revenue} + \sum_{b {\tiny \mbox{ confirmed}}} ({0}- p) - \sum_{b {\tiny \mbox{ unconfirmed}}, \, b > v} \gamma (b-{0}),
\end{equation}
{where a fake bid has value $v = 0$.}
The strict $\gamma$-utility of the miner is 
\begin{equation}
\label{eq:gammaminer2}
\mathcal{U}_m(\mathcal{R}_\gamma), \quad \mbox{where } \mathcal{U}_m(\cdot) \mbox{ is given by \eqref{eq:Um} below}. 
\end{equation}
\end{itemize}
We shall simply refer to \eqref{eq:gammauser} as the user's utility,
and to \eqref{eq:gammaminer1}--\eqref{eq:gammaminer2} as the miner's {\it return} and utility,
i.e., implicitly assume a positive (``strict'') $\gamma >0$. 
\end{definition}

\quad Several remarks are in order.
An honest user's bid satisfies $b=v$, whereas $b\neq v$ corresponds to a strategic user.
A confirmed bid returns to the user
 its true value $v$ (even if $b\neq v$), for which the 
user will pay $p$ (to be specified by the mechanism below). This explains the first summation in \eqref{eq:gammauser}.
Moreover, as mentioned above, 
if a user's bid is not confirmed in the current block,
it is still recorded in the blockchain 
and may very well be confirmed in a future (often near-term) block.
The second summation in \eqref{eq:gammauser} accounts for the cost/penalty
 the strategic user will pay for overbidding ($b>v$). 
 (There's no penalty for underbidding.)



\quad The miner's total return (or net profit) $\mathcal{R}_\gamma$ is specified in the expression in 
\eqref{eq:gammaminer1}, where the first term is the revenue (the miner's receivable, to be specified by the mechanism below);
while the other two terms are the same as the user's return (utility) in  \eqref{eq:gammauser}, which allows for modeling
the miner's possible strategic behavior of injecting fake transactions into the system.  

\quad As to the miner's utility function $\mathcal{U}_m(\mathcal{R}_\gamma)$, the modeling details will be specified in 
the next subsection. 
\subsection{Miner's long-run utility}
\label{sc2}

To properly model the miner's utility, one must recognize and acknowledge that the miner and the 
users differ fundamentally in their {\it time scales}. 
In the context of capturing their strategic behaviors,
users are often more myopic and focus on their short-term gains;
while the miner is concerned with long-term returns, particularly so once we aggregate all the miners into
a single representative miner.
Recall, under the PoS protocol, in addition to the return/profit the miner earns for processing the 
transactions, the miner can also include the return as part of the stake used to earn rewards (in the form of stakes) associated with
validating a block (and handed out separately by the protocol).  
Hence, the more profit the miner receives in each round of processing transactions, 
the more likely the miner will be selected to validate a new block in the future,
which in turn will generate more profit. 
This {\it ``more leads to more''} incentive is simply not present in a user's bidding process.

\quad 
As mentioned earlier, 
the miner's revenue comes from two sources: 
payment from processing the user transactions and the reward from the blockchain.
To earn both, the miner needs to {be} selected by the PoS protocol to mine the block; if not 
the miner gets nothing.
Assume the users' payment $p$ is a constant, which is consistent with our focus on formulating a 
long-run utility for the miner.  
(Not to add, in practice a stationary flow of transactions over a long period 
is expected in a healthy payment ecosystem.)
For $t \ge 1$, let $R_t$ be the number of stakes handed out as reward by the PoS protocol at time $t$.
Let $(M_t, \, t \ge 0)$ be the number of stakes owned by the miner; hence, 
\begin{equation}
\label{eq:Mt}
M_{t} = M_{t-1} + (p +R_{t}) \,1_{\tiny \{ \mbox{the miner is selected at time } t \}}.
\end{equation}
It is important to note that here the payment $p$ may or may not be honest, which allows for 
modeling the miner's strategic behavior.

\quad Now, if the miner is not selected to mine a block, then some other miner will be selected to do so. 
The latter will receive a payment $p_h$, an honest one. This way, the other miner is used as a reference point to 
pin down any possible strategic gain or loss in the utility of the miner we are focusing on.
Let $(N_t, \, t \ge 0)$ be the total number of stakes owned by all the miners. 
We have
\begin{equation}
\label{eq:Nt}
N_{t} = N_{t-1} + (p + R_t) \,1_{\tiny \{ \mbox{the miner is selected at time } t \}}
+  (p_h + R_t) \,1_{\tiny \{ \mbox{the miner not selected at time } t \}},
\end{equation}
where the probability that the miner is selected at time $t$ given the past is
$M_{t-1}/N_{t-1}$.
Denote by $\{\mathcal{F}_t\}_{t \ge 0}$ the filtration generated by the process $(M_t, N_t)$.
Combining \eqref{eq:Mt} and \eqref{eq:Nt}, 
we obtain the dynamics of $(M_t, N_t)$:
\begin{equation}
\label{eq:MNgen}
(M_{t}, N_{t}) \,|\, \mathcal{F}_{t-1} = \left\{ \begin{array}{lcl}
(M_{t-1} + p + R_t, N_{t-1} + p + R_t) & \mbox{with probability}
& M_{t-1}/N_{t-1}, \\
(M_{t-1}, N_{t-1} +p_h + R_t) & \mbox{with probability} & 1-M_{t-1}/N_{t-1}.
\end{array}\right.
\end{equation}
Define the miner's long-run utility by 
\begin{equation}
\label{eq:Um}
\mathcal{U}_m (p): = \liminf_{t \to \infty} \frac{\mathbb{E}M_t}{t}.
\end{equation}
It is worth mentioning that
long-run payoffs (though different) were also used in  \cite{KKKT16, SSZ17} to analyze strategic behaviors in {PoW} blockchain mining. 

\quad To simplify the presentation, below we shall 
assume that the reward $R_t = R \ge 0$ is constant throughout. 
The next theorem derives the miner's long-run utility $\mathcal{U}_m(p)$ using a martingale approach. 
\begin{theorem}
\label{thm:Ump}
Let $\pi_0:= M_0/N_0$ be the miner's initial share,
and $\mathcal{U}_m(\cdot)$ be defined by \eqref{eq:Um}.
Assume that $0 < \pi_0 < 1$.
Then
\begin{equation}
\label{eq:Umformula}
\mathcal{U}_m(p) = 
 \left\{ \begin{array}{lcl}
\pi_0 (p_h + R) & \mbox{for} & p = p_h, \\ 
p +R & \mbox{for} & p > p_h, \\
0 & \mbox{for} & p < p_h.
\end{array}\right.
\end{equation}
\end{theorem}
Refer to Figure \ref{fig:1} for a plot of $\mathcal{U}_m(p)$.

\begin{figure}[h]
\centering
\includegraphics[width=0.55\columnwidth]{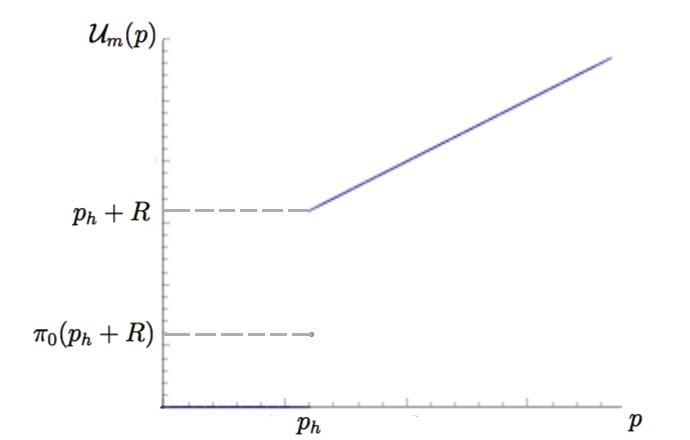}
\caption{Plot of $\mathcal{U}_m(\cdot)$.}
\label{fig:1}
\end{figure}

\begin{proof}
Denote $G:= p+ R$ and $G_h:= p_h + R$,
so the dynamics in \eqref{eq:MNgen} can be expressed as
\begin{equation}
\label{eq:MNcons}
(M_{t+1}, N_{t+1}) \,|\, \mathcal{F}_{t} = \left\{ \begin{array}{lcl}
(M_{t} + G, N_{t} + G) & \mbox{with probability}
& M_{t}/N_{t}, \\
(M_{t}, N_{t} + G_h) & \mbox{with probability} & 1-M_{t}/N_{t}.
\end{array}\right.
\end{equation}
We distinguish two cases: $p = p_h$ and $p \ne p_h$.

Case 1: $p = p_h$ (i.e., the miner is honest).
In this case,
\begin{equation}
\label{eq:M1}
M_{t+1} \,|\, \mathcal{F}_{t}= \left\{ \begin{array}{lcl}
M_{t} + G_h & \mbox{with probability}
& M_{t}/N_{t}, \\
M_{t} & \mbox{with probability} & 1 - M_{t}/N_{t}.
\end{array}\right.
\end{equation}
where $N_t = N_0 + G_h t$ is non-random.
It is easy to see from \eqref{eq:M1} that 
\begin{equation*}
\mathbb{E}M_{t+1} = \left( 1 + \frac{G_h}{N_{t}}\right)  \mathbb{E}M_{t}
= \frac{N_{t+1}}{N_{t}} \mathbb{E}M_{t},
\end{equation*}
which yields
$\mathbb{E}M_t = \pi_0 N_t$ for $t \ge 0$.
(In fact, $(M_t/N_t, \, t \ge 0)$ is a martingale.)
As a result, 
\begin{equation*}
\frac{\mathbb{E}M_t}{t} = \pi_0 \frac{N_t}{t} \longrightarrow \pi_0 G_h, \quad \mbox{as } t \to \infty.
\end{equation*}

Case 2: $p \ne p_h$ (i.e., the miner is strategic).
Set $c: = M_0 +  \frac{G}{G_h - G}N_0$.
It is straightforward from \eqref{eq:MNcons} that
\begin{equation*}
M_{t+1} + \frac{G}{G_h - G} N_{t+1} =  M_{t} + \frac{G}{G_h - G} N_{t} + \frac{GG_h}{G_h - G},
\end{equation*}
 so
 \begin{equation}
 \label{eq:MNlink}
 M_t + \frac{G}{G_h - G} N_t + \frac{GG_h}{G - G_h} t = c,
 \end{equation}
 and
\begin{equation}
\label{eq:M2}
\mathcal{U}_m(p) = \frac{G}{G-G_h} \liminf_{t \to \infty} \frac{\mathbb{E} N_t}{t} + \frac{GG_h}{G_h-G}.
 \end{equation}
By \eqref{eq:MNcons} and \eqref{eq:MNlink}, we obtain the dynamics of $N_t$:
 \begin{equation*}
N_{t+1} = \left\{ \begin{array}{lcl}
N_t + G & \mbox{with probability}
& \left(\frac{G}{G-G_h} N_t + \frac{GG_h}{G_h-G}t + c \right)/N_t, \\
N_t + G_h & \mbox{with probability} & 1-\left(\frac{G}{G-G_h} N_t + \frac{GG_h}{G_h-G}t + c \right)/N_t.
\end{array}\right.
\end{equation*}

\quad First assume that $p > p_h$ (so $G > G_h$). 
Let $L_t: = N_t - G_h t$, and the dynamics of $L_t$ is:
\begin{equation*}
L_{t+1} = \left\{ \begin{array}{lcl}
L_t + (G-G_h) & \mbox{with probability}
& \left(\frac{G}{G-G_h} L_t  + c \right)/(L_t + G_ht), \\
L_t  & \mbox{with probability} & 1-\left(\frac{G}{G-G_h} L_t  + c \right)/(L_t + G_h t).
\end{array}\right.
\end{equation*}
Note that 
\begin{equation}
\label{eq:Ltupper}
\left(\frac{G}{G-G_h} L_t  + c \right)/(L_t + G_h t) \le 1 \Longrightarrow L_t \le (G-G_h)\left(t - \frac{c}{G_h} \right).
\end{equation}
 We have
 \begin{equation}
 \label{eq:Ltineq}
 \begin{aligned}
\mathbb{E}\left(L_{t+1} \,|\, \mathcal{F}_t \right)
& = L_t + \frac{G L_t + c(G-G_h)}{L_t + G_h t} \\
& \ge \left(1 + \frac{1}{t - \frac{G-G_h}{GG_h} c} \right) L_t + \frac{c(G-G_h)/G}{t - \frac{G-G_h}{GG_h} c}.
\end{aligned} 
 \end{equation}
By setting 
\begin{equation*}
I_t: = \frac{L_t + c(G-G_h)/G}{t - \frac{G-G_h}{GG_h} c},
 \end{equation*}
 the inequality \eqref{eq:Ltineq} yields
 $\mathbb{E}(I_{t+1} \,|\, \mathcal{F}_t) \ge I_t$.
 That is, $(I_t, \, t \ge 0)$ is a sub-martingale. 
By \eqref{eq:Ltupper}, we have $\limsup_{t \to \infty} \mathbb{E}(I_t) \le G-G_h$.
The martingale convergence theorem (see \cite[Theorem 4.2.11]{Durrett}) implies that as $t \to \infty$,
\begin{equation*}
I_t \to I_\infty \mbox{ a.s. and in } L^1, \quad \mbox{with } 0 < I_\infty \le G-G_h \mbox{ a.s.}
\end{equation*}
So $L_t/t \to I_\infty$ a.s. and in $L^1$. 
Again by \eqref{eq:Ltineq}, we get
$\lim_{t \to \infty} \mathbb{E}L_{t+1} - \mathbb{E}L_t = \mathbb{E}\left( \frac{G I_\infty}{I_\infty + G_h}\right)$,
which yields 
\begin{equation*}
\mathbb{E}I_\infty = \mathbb{E}\left( \frac{G I_\infty}{I_\infty + G_h}\right).
\end{equation*}
We then have $\mathbb{E}\left(\frac{(G-G_h - I_\infty) I_\infty}{I_\infty + G_h}\right) = 0$,
so $I_\infty = G-G_h$ a.s. 
and $\mathbb{E}L_t/t = G-G_h$.
Therefore, $\mathbb{E} N_t/t \to G$ as $t \to \infty$,
and by \eqref{eq:M2}, 
we get $\mathcal{U}_m(p) =  \frac{G}{G-G_h} G + \frac{GG_h}{G_h-G} = G$ if $p > p_h$. 

\quad Similarly, if $p < p_h$, we can show that $\mathbb{E}N_t/t \to G_h$ as $t \to \infty$,
and hence, $\mathcal{U}_m(p) = \frac{G}{G-G_h} G_h + \frac{GG_h}{G_h-G} = 0$.
\end{proof}

\quad Let's make several remarks. 
First, the formula in \eqref{eq:Umformula} shows that 
the miner's utility $\mathcal{U}_m(\cdot)$ is discontinuous
at the true value $p = p_h$. 
If the miner is strategic so as to be overpaid with 
$p= p_h + \varepsilon$, then
\begin{equation*}
\mathcal{U}_m(p) - \mathcal{U}_m(p_h)=(p+R)-\pi_0 (p_h + R)
= \varepsilon + (1- \pi_0)(p_h + R),
\end{equation*}
i.e., the miner's utility is overshot by a fixed amount $(1 - \pi_0)(p_h + R)$
plus the $\varepsilon$ increment (above $p$). 
Similarly, if $p=p - \varepsilon$, then
\begin{equation*}
\mathcal{U}_m(p) - \mathcal{U}_m(p_h)=0-\pi_0 (p_h + R)
= - \pi_0(p_h + R),
\end{equation*}
i.e., the miner's utility is undershot by a fixed amount $\pi_0(p_h + R)$.
As will be clear in the following section,  
these ``gaps''  (the overshot and the undershot) 
will become a key obstacle to 
apply the (raw) burning second-price auction of \cite{CS23}, in particular when there's a miner-user collusion.

\quad Second, the assumption $0 < \pi_0 < 1$ excludes two extremal cases:
$\pi_0 = 0$ and $\pi_0 = 1$.
If $\pi_0 = 0$, the miner has no stake at hand, hence is effectively a user. 
If $\pi_0 = 1$, then the miner controls the entirely blockchain, an extreme case that rarely happens.
In that case, the miner's utility is always
$\mathcal{U}_m(p) = p+R$, whatever the value of $p$ is, simply because there's no other 
miner as a competitor.

\quad Finally, 
the assumption that $R_t$ is constant may be relaxed. 
If $p = p_h$, the same argument in the proof shows that 
\begin{equation*}
\mathcal{U}_M(p_h) = \pi_0 \left(p_h +  \liminf_{t \to \infty} \frac{\sum_{k = 1}^t R_k}{t} \right).
\end{equation*}
However determining $\mathcal{U}_M(p)$ for $p \ne p_h$ with a time-dependent reward $R_t$ seems to be involved. 

\section{Burning second-price auction}
\label{sc33}

\quad 
Here we start with defining the three desired properties of a transaction fee mechanism.
\begin{definition}
\label{def:UMs}
Let $c\ge 1$ be a positive integer parameter. 
A transaction fee mechanism is said to be: 
\begin{itemize}[itemsep = 3 pt]
\item
user incentive compatible (UIC),
if a user's utility is maximized when biding honestly (i.e., with $b=v$), independent of 
how the other users bid;
\item
miner incentive compatible (MIC), 
if the miner's utility is maximized while honestly implementing the inclusion rule, 
given any bids from the users;
\item
$c$-side contract proofness ($c$-SCP),
if the joint utility of the miner and one or up to $c$ (colluding) users 
is maximized when the miner follows the inclusion rule honestly and the colluding users bid honestly, independent of 
how the other users bid.
\end{itemize}
\end{definition}
While Definition \ref{def:UMs} certainly applies to general utility functions, 
here we focus on the strict $\gamma$-utility in Definition \ref{def:gamut}, with a given value of $\gamma \in (0,1]$. 

\quad Next, we introduce a one-parameter family of the burning second-price auctions,
which naturally extends the burning second-price auction in \cite{CS23}.
\begin{definition}[Burning second-price auction]
\label{def:BSP}
Let $B$ be the block size, and $c\ge 1$ be the (maximum) colluding size.  
Set $\frac{2c}{2c+1} B \le k < B$.
Define the burning second-price auction with parameter $\theta \in (0,1]$ 
by the following set of rules:
\begin{itemize}[itemsep = 3 pt]
\item
Inclusion rule: Choose the $B$ highest bids to include in the block, breaking ties arbitrarily.
Let $(b_1, . . . , b_B)$ denote the included bids ranked in decreasing order. 
\item
Confirmation rule: Randomly select a subset $S \subset \{b_1, . . . , b_k\}$ of size $\left \lfloor{\frac{\theta k}{c}}\right \rfloor$.
All bid in the set $S$ are confirmed, and all other bids $\{b_1, \ldots, b_B\} \setminus S$ are unconfirmed. 
\item
Payment rule: Each confirmed bid pays $p=b_{k+1}$; unconfirmed bids pay nothing.
\item
Miner's revenue: The miner is paid $\theta (b_{k+1} + \cdots +b_{B})$ if $\left \lfloor{\frac{\theta k}{c}}\right \rfloor \ge 1$,
and $0$ otherwise.
Burn the remaining payment collected from the confirmed bids.
\end{itemize}
Let $\texttt{BSP}(\theta)$ denote this transaction fee mechanism. 
\end{definition}

\quad 
First note that a random subset of the included bids are confirmed. 
This can be implemented by trusted on-chain algorithms, e.g. multi-party computation (see \cite{SCW22}).
Second, {if the size of this subset $\left \lfloor{\frac{\theta k}{c}}\right \rfloor = 0$, 
then no transaction is confirmed and the miner is paid nothing (the trivial mechanism).
Hence, we will only need to consider the case $\left \lfloor{\frac{\theta k}{c}}\right \rfloor \ge 1$:}

\begin{equation}
\label{eq:burn}
\left \lfloor{\frac{\theta k}{c}}\right \rfloor b_{k+1} \ge
\frac{\theta k}{2c} b_{k+1} \ge \theta(b_{k+1} + \cdots + b_B),
\end{equation}
where the left side is the total payment collected from the users, and the right side is the miner's revenue. 
{The second inequality in \eqref{eq:burn} follows from the fact that
$b_{k+1} \ge \cdots \ge b_B$ and $k \ge \frac{2c}{2c+1} B$.}
Thus, $\texttt{BSP}(\theta)$ is a valid transaction fee mechanism. 

\quad 
More importantly, whereas it might appear that when $\theta = \gamma$,  
$\texttt{BSP}(\theta)$ reduces to the 
$\texttt{BSP}(\gamma)$ mechanism in \cite{CS23},
there is a crucial difference between the two: While both mechanisms apply the strict $\gamma$-utility to the miner,
as well as to every user, 
the $\texttt{BSP}(\gamma)$ mechanism in \cite{CS23} further applies the same time scale to both miner and user.
In contrast, in the $\texttt{BSP}(\theta)$ mechanism here, a different time scale, {\it long-run average},
is applied to the miner's utility following \eqref{eq:Um},
and more explicitly in \eqref{eq:Umformula}. 
Consequently, we need an extra parameter $\theta$ so as to overcome the discontinuity of the miner's utility 
highlighted at the end of the last section. 
To appreciate this, first note a negative (impossibility) result in the following proposition.

\begin{proposition}
\label{prop:cscp}
$c$-SCP cannot be achieved by $\texttt{BSP}(\theta)$ for any $\theta \in (0,1]$
under the strict $\gamma$-utility for the PoS.
\end{proposition}
\begin{proof}
Without loss of generality, assume $c = 1$. 
Let $(v_1, \ldots, v_B)$ be the honest bids ranked in decreasing order, with $v_k > v_{k+1}$.
Consider the case where the miner and user $k+1$ collude.
Note that user $k+1$'s 
utility is $0$ since $v_{k+1}$ is not confirmed.

\quad Now, suppose user $k+1$ increase the bid to $b_{k+1}:=v_{k+1} + \varepsilon$, with $\varepsilon < v_k - v_{k+1}$.
So $b_{k+1}$ is still unconfirmed.
Then,
user $k+1$'s utility decreases from $0$ to $- \gamma \varepsilon$.
The miner's utility increases from $\pi_0\left(\theta \sum_{\ell = k+1}^B {v}_\ell + R\right)$ to $\theta \left(\sum_{\ell = k+1}^B  {v}_\ell + \varepsilon\right) + R$. 
Thus, the change in the joint utility of the miner and user $k+1$ is
\begin{equation}
\label{eq:changes3}
(1- \pi_0)\left(\theta \sum_{\ell = k+1}^B v_\ell + R \right) + (\theta - \gamma) \varepsilon.
\end{equation}
By making $\varepsilon$ sufficiently small ($\varepsilon \downarrow 0$), 
the first term above will be the
dominant, and thus the change will be positive. 
That is, the miner can collude with user $k+1$ to get a higher joint utility, which
violates $1$-SCP. 
\end{proof}

\quad In the above proof, 
$c$-SCP fails because a small increase in a user's bid
can trigger an overshoot $(1- \pi_0)\left(\theta \sum_{\ell = k+1}^B {v}_\ell + R \right)$
in the miner's utility.
One way to prevent this from happening  
is to put some constraints on $\varepsilon$ (the bid increment), 
and also on the bids $(v_1, \ldots, v_B)$; specifically, a lower bound on $\varepsilon$,
and an upper bound on $\sum_{\ell = k+1}^B v_\ell$. 
Yet even more importantly, we need to impose an upper limit on $\theta$, so as to rule out $\theta - \gamma\ge 0$. 
Because if $\theta \ge \gamma$, then the change to the joint miner-user utility in \eqref{eq:changes3} will always be positive; 
i.e., $c$-SCP will always fail, even with the constraints on bids and bid increments.



\section{Burning second-price auction with minimum tick}
\label{sc4}

\quad Based on the observations following the proof of Proposition  \ref{prop:cscp}, 
we make the following assumption.
\begin{assump}
\label{assump:epsp}
~
\begin{enumerate}[itemsep = 3 pt]
\item[(i)]
The bids are in multiples of $\Delta > 0$, i.e. taking values from the discrete set $\in \{0, \Delta, 2 \Delta, \ldots\}$.
\item[(ii)]
Let $(v_1, \ldots, v_B)$ be honest bids ranked in decreasing order.  
Then, there exists a constant $\kappa > 0$ such that $\sum_{\ell = k+1}^B v_\ell \le \kappa$.
\end{enumerate}
\end{assump}

\quad The assumption in (i) imposes a (minimum) bid increment $\Delta > 0$, or ``tick'' size, 
which is common in practice (e.g., auction, stock markets, etc).
It is a block validity rule, so is part of the inclusion rule of the transaction fee mechanism.
The assumption (ii) gives an upper bound for the miner's revenue 
$\theta \sum_{\ell = k+1}^B v_\ell$ if all participants are honest.
A sufficient condition is that
\begin{equation*}
\sum_{\ell = 1}^B v_\ell \le \kappa,
\end{equation*}
i.e. the total amount of top (honest) bids is bounded.
This is related to the idea of {\em bounded rationality} \cite{Simon90}, 
where the users bid rationally.
So the bid flow is expected to stabilize, or close to stationarity. 
In practice, the upper bound $\kappa$ can be inferred or estimated from 
the bid-stream data.

\quad The next theorem confirms that under Assumption \ref{assump:epsp},
$\texttt{BSP}(\theta)$ satisfies UIC, MIC and $c$-SCP
for a suitable choice of $\theta$.
\begin{theorem}
\label{thm:main}
Let Assumption \ref{assump:epsp} hold. Then,
for $0 < \pi_0 < 1$, $\Delta > (1 - \pi_0) R/\gamma$, 
and
\begin{equation}
\label{eq:thetarange}
\theta \le \overline{\theta}: = \min \left( \frac{\pi_0 R}{(1 - \pi_0) \kappa}, \,  
\frac{\gamma \Delta - (1 - \pi_0) R} {(1- \pi_0) \kappa + \Delta} \right),  
\end{equation}
$\texttt{BSP}(\theta)$ achieves UIC, MIC and $c$-SCP under the strict $\gamma$-utility for the PoS.
\end{theorem}

\quad Several remarks are in order before we prove Theorem \ref{thm:main}.
First, observe the dependence of $\overline{\theta}$ on the parameters
$(\pi_0, R, \Delta, \kappa)$ as specified below, and 
also refer to Figure \ref{fig2} for illustration: 
\begin{itemize}[itemsep = 3 pt]
\item
$\overline{\theta}$ is increasing in $\Delta$,
and $\overline{\theta} \to 0$ when $\Delta \to (1 - \pi_0) R/\gamma$,
and $\overline{\theta} \to \min \left(\frac{\pi_0 R}{(1 - \pi_0) \kappa}, \gamma \right)$ 
when $\Delta \to \infty$.
\item
$\overline{\theta}$ is decreasing in $\kappa$,
and $\overline{\theta} \to \gamma - (1 - \pi_0) R/\Delta$ when $\kappa \to 0$,
and $\overline{\theta} \to 0$ when $\kappa \to \infty$.
\item
$\overline{\theta}$ is increasing in $\pi_0$,
and $\overline{\theta} \to 0$ when $\pi_0 \to 0$,
and $\overline{\theta} \to \gamma$ when $\pi_0 \to  1$.
\item
$\overline{\theta}$ increases from $0$ to $\frac{\pi_0 \gamma \Delta}{(1 - \pi_0)\kappa + \pi_0  \Delta}$
when $R \in \left[0, \frac{(1 - \pi_0) \kappa\gamma  \Delta}{(1 - \pi_0)\kappa + \pi_0 \Delta}\right]$;
and then decreases 
to $0$
when $R \in \left(\frac{(1 - \pi_0) \kappa\gamma \Delta}{(1 - \pi_0)\kappa + \pi_0 \Delta}, 
\frac{\gamma \Delta}{1 - \pi_0}\right).$
\end{itemize}

Second, as we will see in the proof below, 
$\texttt{BSP}(\theta)$ satisfies UIC for any $\theta$,
and satisfies MIC for any $\theta \le \gamma$. 
(Note that $\theta\le \overline{\theta}\le \gamma$ as observed above.)
The condition in \eqref{eq:thetarange} is required for $c$-SCP,
where the bound $\frac{\pi_0 R}{(1 - \pi_0) \kappa}$
is to offset the undershoot gap in the miner's utility,
while the bound 
$\frac{\gamma \Delta - (1 - \pi_0)R}{(1 - \pi_0) \kappa + \Delta}$
is to offset its overshoot gap.

\begin{figure}[h]
     \begin{subfigure}[b]{0.45\textwidth}
     \centering
         \includegraphics[width=0.85\textwidth]{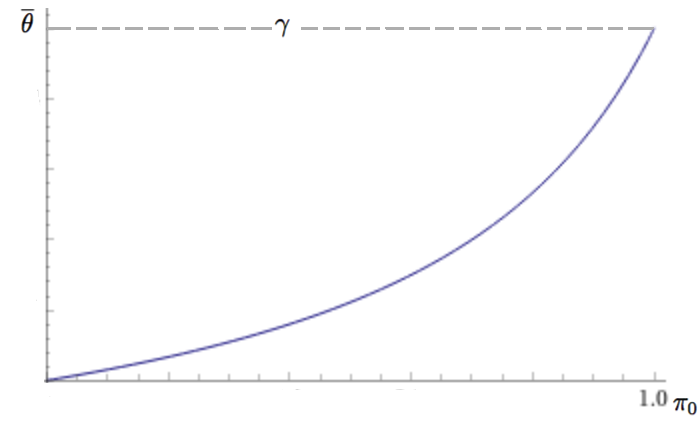}
         \caption{Plot of $\pi_0 \to \overline{\theta}$.}
     \end{subfigure}
     \hfill
     \begin{subfigure}[b]{0.45\textwidth}
     \centering
         \includegraphics[width=0.85\textwidth]{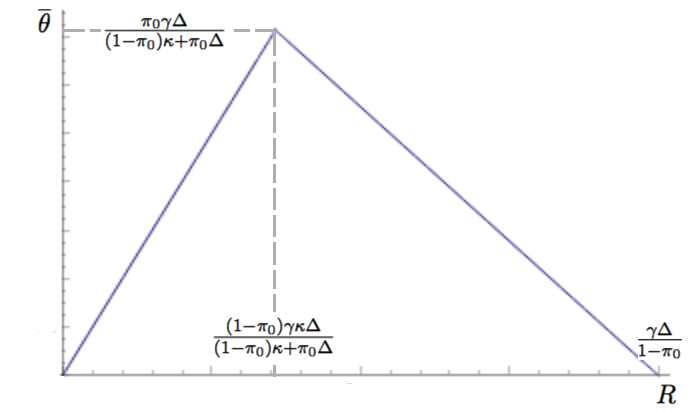}
         \caption{Plot of $R \to \overline{\theta}$.}
     \end{subfigure}
     \hfill
     \begin{subfigure}[b]{0.45\textwidth}
     \centering
         \includegraphics[width=0.85\textwidth]{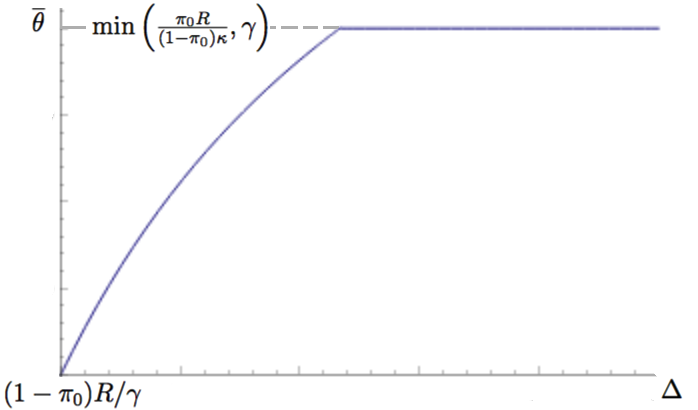}
         \caption{Plot of $\Delta \to \overline{\theta}$.}
     \end{subfigure}
     \hfill
     \begin{subfigure}[b]{0.45\textwidth}
     \centering
         \includegraphics[width=0.85\textwidth]{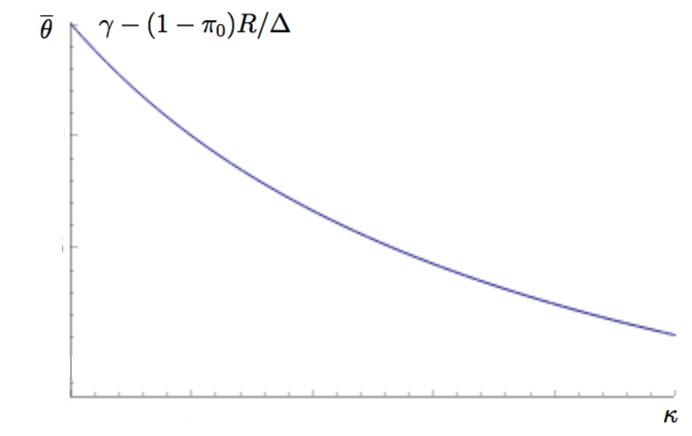}
         \caption{Plot of $\kappa \to \overline{\theta}$.}
     \end{subfigure}
        \caption{Plots of $\overline{\theta}(\pi_0, R, \Delta, \kappa)$.}
        \label{fig2}
\end{figure}

\quad Following Theorem \ref{thm:main}, we know that
$\texttt{BSP}(\theta)$ achieves UIC, MIC and $c$-SCP
as long as $\theta \le \overline{\theta}$. 
Recall $\texttt{BSP}(\theta)$ confirms a $\theta$-portion of the top $k$ bids; hence, 
 it is desirable to take $\theta = \overline{\theta}$ so as to attain the best efficiency.  
After all, the blockchain is used to confirm transactions,
so the more it confirms in each block, the more efficient the mechanism is.

\quad As indicated in Section \ref{sc31}, 
Theorem \ref{thm:main} can be readily extended to the miner-miner collusion.
To do so, it suffices to replace $\pi_0$ with $\pi^{\tiny \mbox{gr}}_0$ in \eqref{eq:thetarange}, 
where $\pi^{\tiny \mbox{gr}}_0$ is the total share of the miner collusion group.
In this case it is possible, albeit highly unlikely, that $\pi_0^{\tiny \mbox{gr}} = 1$
(i.e., a group of miners control the entire blockchain).
Then $\mathcal{U}_m(p) = p+R$ for all $p$;
refer to the earlier remarks after Theorem \ref{thm:Ump}.
In this case, the condition in \eqref{eq:thetarange} reduces to
$\theta \le \overline{\theta} = {\gamma}$, 
{and the raw $\texttt{BSP}(\gamma)$ achieves all the desired properties.} 
 

\medskip

\quad We are now ready to prove Theorem \ref{thm:main}. We split 
the proof into three separate lemmas, each proving one of the three properties UIC, MIC and $c$-SCP.
This split allows us to bring out the specific conditions required for each property;
in particular, note that the range of $\theta$ in \eqref{eq:thetarange} can be relaxed (extended) for UIC and MIC,
and Assumption \ref{assump:epsp} is only needed for establishing $c$-SCP.

\begin{lemma}[UIC]
\label{lem:UIC}
For any $\theta > 0$, $\texttt{BSP}(\theta)$ achieves UIC.
\end{lemma}

\begin{proof}
Any user $i$ who submits an honest bid, i.e., $b_i=v_i$, 
obtains a utility of $v_i-{v_{k+1}} $, with {$p = v_{k+1}$} being the payment, if the bid is confirmed;
or a utility of $0$, if the bid is unconfirmed. 

\quad Now, suppose user $i$ changes the bid to $b_i\neq v_i$. 
{There are two cases:}

\smallskip\noindent
{Case 1: $i \le k$.
If $b_i$ is among the top $k$ bids, 
the utility remains unchanged at $v_i - v_{k+1}$ if the bid is confirmed;
the utility is $-\gamma (b_i-v_i)^+ \le 0$ if the bid is unconfirmed.
In both cases, the utility cannot exceed that of the honest bid.
If $b_i$ drops out of the top $k$ bids, 
the utility is $0$ -- no greater than the original utility.
}

\smallskip\noindent
{Case 2: $i > k$.
If $b_i$ is not among the top $k$ bids,
the bid is unconfirmed and the utility 
is $- \gamma(b_i - v_i)^+ \le 0$.
If $b_i$ joins the top $k$ bids, 
the payment changes to $p = v_k$.
The user's utility is $v_i - v_k \le 0$ if the bid is confirmed,
and is $-\gamma(b_i - v_i)^+ = -\gamma (b_i - v_i) \le 0$ 
if the bid is unconfirmed.}

 \quad The above also covers the case of injecting a fake transaction, i.e., one with $v_i=0$ and $b_i>0$. 
\end{proof}

\begin{lemma}[MIC]
\label{lem:MIC}
For any $\theta \le \gamma$, $\texttt{BSP}(\theta)$ achieves MIC.
\end{lemma}
\begin{proof}
The miner can deviate from being honest by 
either not including the highest $B$ bids,
or injecting fake bids. 
Recall, given the bid vector $(b_1, \ldots, b_B)$ in decreasing order,
the miner's utility
$\mathcal{U}_m(\cdot)$ is non-decreasing in the miner's receivable (i.e., revenue) $\theta \sum_{\ell = k+1}^B b_\ell$. 
Hence, it is impossible for the miner not to include the highest $B$ bids. 

\quad As to injecting fake bids, this amounts to the miner changing some bid $b_i$ to $f$ (a fake one). 
Again, as the miner's utility
$\mathcal{U}_m(\cdot)$ is non-decreasing in the miner's total return
(revenue minus payment/cost), it suffices to examine the changes to the miners total return. 

There are two cases:

\smallskip\noindent
Case 1: Suppose $f$ is confirmed (hence, $f$ must be among the top $k$ bids). 
\begin{itemize}
\item[(1a)] If $i\le k$, then $f$ replaces $b_i$ in the original set of $k$ bids.
Thus, the miner needs to pay $b_{k+1}$ for the confirmed (fake) bid, while the miner's own
revenue remains unchanged. So, the net change to the miner's total return is $-b_{k+1} <0$.
\item[(1b)] if $i > k$, then $f$ joins the set of top $k$ bids, replacing $b_k$, which in turn replaces $b_{k+1}$ as payment;
and this is also the miner's payment for the confirmed (fake) bid. Accordingly, the change to the miner's
revenue is $\theta (b_k-b_i)$. Thus, the net change is $-b_k +\theta (b_k -b_i) < 0$, as $\theta\le 1$.
\end{itemize}

\smallskip\noindent
Case 2: Suppose $f$ is not confirmed. The fake bid costs the miner zero payment,  but incurs a negative amount
 $-\gamma f <0$ (since the fake bid has zero value). In addition, there are changes to the miner's revenue:
\begin{itemize}
\item[(2a)] Suppose $f$ is among the top $k$ bids. If $b_i$ is also among the top $k$ bids, then, there's no change to the miner's revenue.
If $b_i$ is not among the top $k$ bids, i.e., $i>k$, then similar to (1b) above, the change to the miner's
revenue is $\theta (b_k-b_i)$. Thus the net change (to the miner's total return) 
is $-\gamma f  +\theta (b_k -b_i)= (\theta b_k-\gamma f) - \theta b_i < 0$, since $\theta\le \gamma$ and $b_k\le f$. 
\item[(2b)] Suppose $f$ is not among the top $k$ bids. 
 If $b_i$ is among the top $k$ bids, then replacing
$b_i$ by $f$ will result in moving $b_{k+1}$ into the top $k$ set.
 The change to the miner's revenue is $\theta (f-b_{k+1})$; hence, the net change is $-\gamma f  +\theta (f -b_{k+1}) < 0$,
since  $f \le b_{k+1}$. 
If $b_i$ is also not among the top $k$ bids, then replacing
$b_i$ by $f$, the change to the miner's revenue is $\theta (f-b_i)$; hence, the net change is 
$-\gamma f  +\theta (f -b_i)=-(\gamma-\theta) f  -\theta b_i < 0$
since $\theta\le \gamma$. 
\end{itemize}
Since the miner's total return will be (strictly) reduced in all cases, the miner will have no incentive to inject any fake bids.
\end{proof}

\begin{lemma}[$c$-SCP]
\label{lem:cSCP}
Under the assumptions in Theorem \ref{thm:main},
$\texttt{BSP}(\theta)$ achieves $c$-SCP. 
\end{lemma}
\begin{proof}
A collusion between a miner and the users may  one of the three strategies:
the miner does not include the highest bids;
the miner injects fake bids;
some users bid untruthfully. 
Denote by $C$ the set of the colluding users, with cardinality  $|C| \le c$.
Equivalently, assume that the miner and the users collude in the following order:
\begin{itemize}
\item Step 1.
The miner includes bids strategically. 
This is the same as the miner deleting the (real) bids one by one, and then including the highest bids.
\item Step 2.
The miner replaces some real bids (not in the set $C$) with fake bids.
\item Step 3.
Some users in $C$ change their bids untruthfully.
\end{itemize}

\smallskip
Step 1. Let $(v_1, \ldots, v_B, \ldots)$ denote the honest bids ranked in decreasing order.
Suppose $v_i$ with $i > k+1$ is deleted. 
This will not affect either the miner's utility or any user's utility.

\quad 
Suppose $v_i$ with $i \le k+1$ is deleted. 
  Then, the miner's
revenue changes from $\theta \sum_{\ell = k+1}^B v_\ell$ to $\theta \sum_{\ell = k+2}^{B+1} {v}_\ell$.
This change will be a net decrease, in which case
the miner's utility drops to $0$ and stays at $0$;
 unless ${v}_{k+1}= {v}_{B+1}$, in which case there's no change to the miner's revenue. 

On the other hand, any user's utility can increase by at most $\frac{\theta}{c} (v_{k+1} - v_{k+2})$. 
If the miner deletes one more bid, any user's utility can increase by at most $\frac{\theta}{c} (v_{k+2} - v_{k+3})$,
accumulating to $\frac{\theta}{c} (v_{k+1} - v_{k+3})$, and so forth.
Thus, any user's utility can increase by at most $\frac{\theta}{c} v_{k+1}$. 
Since there are at most $c$ colluding users, 
the joint utility of the colluding users can increase by at most $\theta v_{k+1}$.
So the change in the joint utility of the miner and the colluding users is at most
\begin{equation}
\label{eq:joint1}
-  \pi_0 \left(\theta \sum_{\ell = k+1}^B v_\ell+ R\right) + \theta v_{k+1} 
\le (1 - \pi_0) \theta \kappa - \pi_0 R \le 0,
\end{equation}
since $v_{k+1} \le \sum_{\ell = k+1}^B v_\ell \le \kappa$ by Assumption \ref{assump:epsp},
and $\theta \le \frac{\pi_0 R}{(1 - \pi_0) \kappa}$ by \eqref{eq:thetarange}.


\medskip
Step 2.
From the proof of Lemma \ref{lem:MIC}, it is shown that replacing a (real) bid with a fake one cannot increase the miner's utility in any scenario. Furthermore, it is clear that no user's utility can increase if the fake bid $f$ is 
among the top $k$ bids, as this may increase the payment (from $v_{k+1}$ to $v_k$). 
Hence, we only need to consider case (2b) in the proof of Lemma \ref{lem:MIC}.
In that case, the net change to the miner's total return (revenue minus cost) is $<0$; hence
the miner's utility will drop to $0$ and stay at $0$.  
Moreover, following the argument in Step 1 above, 
a user's utility can increase by at most $\frac{\theta}{c}\left(v_{k+1} - \max\{f, v_{k+2}\} \right)$,
where $\max\{f, v_{k+2}\}$ represents the $(k+1)$-th bid (i.e. the largest unconfirmed bid) after the replacement.
If the miner replace one more bid by another fake $f'$, 
any user's utility can further increase by at most $\frac{\theta}{c}\left(\max\{f, v_{k+2}\} - \max\{f', \min(f, v_{k+2})\} \right)$,
where the $\max$ term is the $(k+1)$-th bid after replacement. 
Repeating this argument will lead to the same 
inequality in \eqref{eq:joint1}.

\medskip
Step 3. {Recall that} $(v_1, \ldots, v_B)$ are the honest bids ranked in decreasing order. 
Without loss of generality, assume that the colluding users change their bids in an ascending order; 
i.e. the one with the lower true value changes first. 
Suppose user $i$ replaces $v_i$ with $b_i \ne v_i$.

\smallskip
Case 1: $i > k$. 

\begin{itemize}
\item[(1a)] 
$b_i$ is among the top $k$ bids. In this case, we have $v_k\ge v_i$ and $b_i> v_k$; moreover, 
$v_k$ becomes the $(k+1)$-th bid (i.e. the largest unconfirmed bid). 
Then, user $i$'s utility changes from $0$ to $v_i-v_k \le 0$ if $b_i$ is confirmed, and to $-\gamma (b_i - v_i) <0$ if unconfirmed. 
Thus, user $i$'s utility change is 
{
\begin{equation}
\label{user1a}
\begin{aligned}
-\left(1 - \frac{\theta}{c} \right) \gamma (b_i - v_i) - \frac{\theta}{c} (v_k- v_i) 
&\le -\left(1 - \frac{\theta}{c} \right) \gamma  \left (v_k- v_i \right)  - \frac{\theta}{c} \gamma (v_k - v_i)  \\
&= -\gamma (v_k-v_i),
\end{aligned}
\end{equation}
}
with the inequality following from $b_i >v_k \ge v_i$ and {$\gamma \le 1$}.

\noindent
The miner's revenue changes from $\theta \sum_{\ell = k+1}^B v_\ell$ to $\theta \left(\sum_{\ell = k}^B v_\ell - v_i\right)$.
So, the change to the miner's utility is 
\begin{equation}
\label{miner1a}
\pi_0 \left( \theta \sum_{\ell = k+1}^B v_\ell+ R \right) \longrightarrow \theta \left(\sum_{\ell = k}^B v_\ell - v_i\right) + R.
\end{equation}
Thus, the change to the joint utility of the miner and the colluding user is at most
\begin{eqnarray}
\label{eq:joint2}
&&(1 - \pi_0)\left(\theta \sum_{\ell = k+1}^B e_\ell + R \right)- {(\gamma - \theta)} (v_k - v_i) \nonumber\\
&\le& (1 - \pi_0)(\theta \kappa + R) - {(\gamma - \theta)} \Delta \le 0,
\end{eqnarray}
where the first inequality follows from $\sum_{\ell = k+1}^B e_\ell \le \kappa$ and $v_k-v_i \ge \Delta$ by Assumption \ref{assump:epsp},
and the second inequality follows from $\theta \le \frac{\gamma \Delta - (1 - \pi_0)R}{(1 - \pi_0) \kappa + \Delta}$ by \eqref{eq:thetarange}.

\smallskip

\item[(1b)]
$b_i$ is not among the top $k$ (hence, unconfirmed).
The miner's revenue changes from $\theta \sum_{\ell = k+1}^B v_\ell$ to $\theta \left(\sum_{\ell = k+1}^B v_\ell + b_i - v_i \right)$.
For the users who are among the top $k$, 
let their payment be $v'_{k+1}$ (if confirmed).
There are two possibilities.
\begin{itemize}
\item
If $b_i > v_i$ (overbidding), then $v'_{k+1} \ge v_{k+1}$. 
Then, user $i$'s utility decreases from $0$ to $-\gamma(b_i - v_i)<0$, 
whereas all other users' utilities do not increase (as their payment my increase).
The miner's utility increases, and the change is the same as in \eqref{miner1a}. 
Thus, the change in the joint utility is at most 
$(1- \pi_0) \left( \theta \sum_{\ell = k+1}^B e_\ell + R\right) - ( \gamma -\theta) (b_i - v_i) \le 0$
for the same reason as in \eqref{eq:joint2}.
\item
If $b_i < v_i$ (underbidding), then $v'_{k+1} \le v_{k+1}$.
So user $i$'s utility is unchanged (at $0$);
the utility of any other user among the top $k$ increases by at most $\frac{\theta}{c} \left(v_{k+1} - v'_{k+1} \right)$;
all other users are unaffected. 
Thus, the joint utility of the colluding users can increase by at most $\theta (v_{k+1} - v'_{k+1}) \le \theta v_{k+1}$;
while the miner's utility drops to $0$.
This reduces to the scenario in \eqref{eq:joint1},
and hence the joint utility of the miner and the colluding users cannot increase.
\end{itemize}

\end{itemize}

\smallskip
Case 2: $i \le k$.

\begin{itemize}
\item[(2a)] $b_i$ remains among the top $k$.
So user $i$'s utility change is equal to $0$ if confirmed, and equal to 
$-\gamma(b_i - v_i)^+ < 0$ if unconfirmed. 
All other users' utilities and the miner's utility remain unchanged. 
So it is obvious that the joint utility of the miner and the colluding users cannot increase.

\item[(2b)] $b_i$ drops out of the top $k$, which means
$b_i \le v_{k+1}$.
Then, the miner's utility either remains unchanged (if $b_i = v_{k+1}$) or drops to $0$ (if $b_i < v_{k+1}$). 
%
The utility of the colluding user(s) can increase by at most 
$\theta(v_{k+1} - \max(b_i, v_{k+2})) \le \theta v_{k+1}$, following the same argument as in Step 2.
So, this leads again to the scenario in \eqref{eq:joint1}, 
and hence the joint utility of the miner and the colluding users cannot increase.

\end{itemize}

\medskip
\quad Finally, consider the case of further replacements  by the colluding users. 
\begin{itemize} 
\item[(3a)]
Suppose the miner's revenue increases 
from $\theta \sum_{\ell = k+1}^B v_\ell$ to 
$\theta \left(\sum_{\ell = k+1}^B v_\ell + \varepsilon^{\uparrow}\right)$, with $\varepsilon^\uparrow > 0$.
A closer inspection of the previous analysis on cases when the miner's revenue increases, case (1a) and case (1b) with overbidding,
implies that the joint utility of the colluding users will decrease by at least 
$\gamma \varepsilon^\uparrow$.
Thus, the change to the joint utility of the miner and the colluding users is at most
\begin{equation*}
(1 - \pi_0) \left( \theta \sum_{\ell = k+1}^B v_\ell + R \right) + {(\theta - \gamma)} \varepsilon^{\uparrow} \le 0, 
\end{equation*}
for the same reason as in \eqref{eq:joint2}, taking into account that $\varepsilon^\uparrow > \Delta$.
\item[(3b)]
Suppose the miner's revenue decreases from $\theta \sum_{\ell = k+1}^B {v}_\ell$ 
to $\theta \left(\sum_{\ell = k+1}^B v_\ell - \varepsilon^{\downarrow}\right)$, with $\varepsilon^{\downarrow} > 0$.
Then, the miner's utility drops 
 to $0$.
From the previous analysis on cases when the miner's revenue decreases, case (1b) with underbidding and case (2b),
we know the joint utility of the colluding users increases by at most $\theta \varepsilon^{\downarrow}$.
So the joint utility of the miner and the colluding users is at most
\begin{equation*}
- \pi_0 \left(\theta \sum_{\ell = k+1}^B v_\ell  + R \right) + \theta \varepsilon^{\downarrow} \le 0, 
\end{equation*}
for the same reason as in \eqref{eq:joint1}. (Clearly, $\varepsilon^{\downarrow} \le \sum_{\ell = k+1}^B v_\ell$.)
\end{itemize}

\end{proof}


\section{Conclusions}
\label{sc5}

\quad In this paper, we consider the transaction fee mechanism design for the PoS protocol.
Motivated by the miner's long-term objective, 
we propose a long-run utility for the miner, and demonstrate that it exhibits discontinuity.
While the raw burning second price-auction recently proposed in the TFM literature fails to satisfy $c$-SCP under this long-run utility, 
a one-parameter generalization along with a minimum tick size
achieves all three desired properties (UIC, MIC and $c$-SCP) for TFM.

\quad There are several directions to extend this work. 
An obvious one is to consider the TFM for the PoW protocol,
where the computational cost (and the R\&D cost) needs to be taken into account.
Second, one can also consider other types of strategic behaviors
such that splitting a bid into small ones.
This resorts to the study a combinatorial auction.
Finally, some recent work proposed the notion of maximal extractable value (MEV) 
in the context of Flash Boys \cite{DG20}, and the proposer-builder separation (PBS) solution \cite{BV21}.
It would be interesting to analyze the TFM in these contexts.

\bigskip
{\bf Acknowledgements:} 
We thank Elaine Shi for introducing and explaining to us the strict $\gamma$-utility.
We thank Aviv Yaish for pointing out the reference \cite{GY22}.
W.\ Tang gratefully acknowledges financial support through NSF grants DMS-2113779 and DMS-2206038,
and through a start-up grant at Columbia University.
David Yao's work is part of a Columbia-CityU/HK collaborative project that is supported by the InnoHK Initiative, The
Government of the HKSAR and the AIFT Lab.

\bibliographystyle{abbrv}
\bibliography{unique}
\end{document}